%% file: main.tex
\documentclass[letterpaper]{article} 
\usepackage{main}  
\usepackage{times}  
\usepackage{helvet}  
\usepackage{courier}  
\usepackage[hyphens]{url}  
\usepackage{graphicx} 
\usepackage{natbib}  
\usepackage{caption} 
\usepackage[utf8]{inputenc}
\usepackage[english]{babel}
\usepackage{amssymb,amsmath,amsthm}
\usepackage{subfig}
\usepackage{graphicx}
\usepackage{amsfonts,xcolor,colortbl} 
\usepackage{xspace}

\newtheorem{theorem}{Theorem}
\newtheorem{definition}[theorem]{Definition}

\newtheorem{lemma}[theorem]{Lemma}
\newtheorem{proposition}[theorem]{Proposition}

\newcommand{\clip}{\mathop{\textrm{clip}}}
\newcommand{\sign}{\mathop{\textrm{sign}}}

\usepackage{algorithm}
\usepackage{algpseudocode}

\newcommand{\uvmeaniter}{\textsc{UVMRec}}
\newcommand{\uvmeanstep}{\textsc{UVM}}

\newcommand{\stratcoin}{\textsc{StratCoinpress}}

\newcommand{\E}{\mathbb{E}}

\newcommand*{\eg}{e.g.,\xspace}
\newcommand*{\ie}{i.e.,\xspace}
\newcommand*{\Ie}{I.e.,\xspace}

\usepackage{enumitem}

\usepackage{booktabs}

\usepackage{newfloat}
\usepackage{listings}
\DeclareCaptionStyle{ruled}{labelfont=normalfont,labelsep=colon,strut=off} 
\floatstyle{ruled}
\newfloat{listing}{tb}{lst}{}
\floatname{listing}{Listing}
\title{A Simple and Practical Method for Reducing\\the Disparate Impact of Differential Privacy}
\author{
    Lucas Rosenblatt\textsuperscript{\rm 1}, Julia Stoyanovich\textsuperscript{\rm 1}, Christopher Musco\textsuperscript{\rm 1}
}
\affiliations{
    \textsuperscript{\rm 1}New York University\\

    \{lucas.rosenblatt, stoyanovich, cmusco\}@nyu.edu
}

\usepackage{bibentry}

\begin{document}

\maketitle

\begin{abstract}
Differentially private (DP) mechanisms have been deployed in a variety of high-impact social settings (perhaps most notably by the U.S. Census).
Since all DP mechanisms involve adding noise to results of statistical queries, they are expected to impact our ability to accurately analyze and learn from data, in effect trading off privacy with utility. Alarmingly, the impact of DP on utility can vary significantly among different sub-populations.
A simple way to reduce this disparity is with \emph{stratification}. First compute an independent private estimate for each group in the data set (which may be the intersection of several protected classes), then, to compute estimates of global statistics, appropriately recombine these group estimates. Our main observation is that naive stratification often yields high-accuracy estimates of population-level statistics, without the need for additional privacy budget.
We support this observation theoretically and empirically.
Our theoretical results center on the private mean estimation problem, while our empirical results center on extensive experiments on private data synthesis to demonstrate the effectiveness of stratification on a variety of private mechanisms. Overall, we argue that this straightforward approach provides a strong baseline against which future work on reducing utility disparities of DP mechanisms should be compared.
\end{abstract}

\section{Introduction}
\label{sec:intro}
Two moral and legal imperatives, data privacy and algorithmic equity, have received significant recent research attention. For ensuring data privacy, differential privacy (DP) has emerged as a gold standard technique \citep{dwork2014algorithmic}.\footnote{We will use ``DP'' to mean ``differential privacy'' or ``differentially private,'' depending on the context.} Academics and practitioners alike use DP algorithms to solve problems with sensitive data; notably, big tech companies like Google, Microsoft and Apple rely on differential privacy for protecting customer data \citep{erlingsson2014rappor, ding2017collecting, cormode2018privacy}. One recent high-profile use of DP was in the 2020 United States Census, which includes statistical disclosures with data for millions of Americans \citep{bureau2021disclosure,christ2022differential,groshen2022disclosure,hawes2020implementing}.
Additionally, a wide variety of popular algorithms have DP versions, from fundamental statistical methods \citep{dwork2014algorithmic}, to machine learning \citep{ji2014differential} and even deep learning \citep{abadi2016deep} techniques.

However, DP can be in tension with algorithmic equity, where, unlike for privacy, there is no ``gold standard'' definition, and where potential harms are murky and plentiful \citep{corbett2018measure, mitchell2021algorithmic}. The 2020 Census, for example, was criticized over concerns that DP techniques would have adverse impacts on estimating demographic proportions \citep{ruggles2019differential} and on redistricting \citep{kenny2021use}. In fact, much attention has been paid lately to disparities in performance of DP mechanisms in a variety of settings, and to potential harms resulting from such disparities~\cite{fioretto2022differential}.  Many of these works demonstrate disparities in performance for different demographic groups within the data~\cite{bagdasaryan2019differential,ganev2022robin}. In-line with recent literature, we broadly refer to such disparities as the \emph{disparate impact} of DP.

As a concrete example, consider the problem of DP data synthesis. Suppose we have a data set $X$ that can be split into $k$ disjoint groups $G_1 \cup \ldots \cup G_k = X$, each of which might contain individuals in the intersection of several protected classes. A concern could be that a private data synthesis method run on $X$ does not faithfully represent the data distribution for some of these groups, which can lead to technical bias~\cite{DBLP:journals/tois/FriedmanN96} --- in the sense of disparities in accuracy, appropriately measured --- when machine learning (ML) models or statistics are fit to the synthetic data \cite{ganev2022robin}.

\paragraph{Contributions}
\label{sec:contributions}
With the above example in mind, our paper analyzes a baseline approach to address the disparate impact of DP mechanisms: \textit{stratification}. We note that general stratification-based methods are fundamental in statistics and may already be used by practitioners in privatizing data; however, we do not know of any work that \textit{formally employs stratification to address the disparate impact of DP}. In particular, we could simply run a DP mechanism separately on $G_1, \ldots, G_k$ and report the results. If we have access to publicly available estimates of the size of each group, we could then take a weighted combination of the results to get a statistical estimate, or to generate DP synthetic data, for the global population. For example, we might fit data synthesizers $D_1, \ldots, D_k$ for each group.  Then, to obtain synthetic data for the entirety of $X$, we would sample from synthesizer $D_i$ with probability proportional to $|G_i|/\sum_{i=1}^k |G_i|$. 

Intuitively, stratification minimizes disparate impact: for each group, we represent the data as well as we would have if the other groups did not exist.\footnote{Some DP synthesizers based on the ``\textit{Select, Measure, Project}'' approach, like MWEM \cite{hardt2012simple}, use stratification in an \emph{implicit  way} because they take measurements of data based on marginal queries. However, these algorithms often rely on mutual information or other metrics to select measurements that are most informative about the joint distribution; in doing so, they may not measure marginals with respect to \emph{all} attributes in a data set. This can leave groups vulnerable to disparate impact, and in fact, seems the likely culprit in the observed disparate impacts by \cite{ganev2022robin} and others.} So, the main question we ask is whether or not this simple approach can lead to high-quality estimates of \emph{global} population statistics. Is something lost by  treating individual groups separately? We make progress on answering this question, arguing that the cost of stratification is small
or even negligible.
Specifically, we make the following contributions:

    (1) We validate the stratification strategy by first considering the problem of mean estimation. By making Dirichlet assumptions on the prior distribution of group sizes, we develop a theoretical understanding of the impact of stratification on the population mean estimate, and show that this impact is limited. Furthermore, for some state-of-the-art adaptive mean estimation techniques like \textsc{Coinpress} \citep{biswas2020coinpress}, stratification can even \emph{reduce error in estimating the global population mean}, while also giving estimates for each stratum that are as accurate as one can hope for.
    
    (2) Next, we validate our strategy on stratified DP data synthesizers, motivated by prior work highlighting the disparate impacts of these algorithms \citep{bagdasaryan2019differential,ganev2022robin}. Our results indicate that stratification leads to \emph{minimal overall utility loss for synthetic data} in practical privacy regimes, while also \emph{reducing disparities} in utility across subgroups in the data.

\subsection{Related Works}
\paragraph{Disparate impact of DP} Informally, a DP mechanism exhibits \textit{disparate impact} when it leads to adverse outcomes for historically disadvantaged (\ie protected) population groups, even if the mechanism appears neutral or unbiased on its face. Legally, a practice that adversely impacts protected groups can be considered discriminatory even without obvious categorization or intent to harm \cite{garrow2014toward, feldman2015certifying,barocas2016big}. 

Prior work examining DP mechanisms found concerning disparate impact and other fairness trends. Some of this work has focused on bias introduced by private stochastic gradient descent (DP-SGD) \cite{abadi2016deep,mironov2017renyi}: Empirically, \citet{bagdasaryan2019differential} discovered that DP-SGD amplifies noise in the data and adversely impacts certain subgroups, while theoretically, \citet{tran2021differentially} showed that multiplicative effects on the Hessian loss in DP-SGD affect the proximity of group-specific data to the decision boundary. Others have also investigated DP-SGD's effects on image generation tasks \cite{cheng2021can} and proposed adaptive clipping mechanisms to reduce negative subgroup impacts \cite{xu2021removing}, or suggested that different DP mechanisms for deep learning have reduced disparate impact~\cite{uniyal2021dp}. 
Other works assessing the fairness impacts of DP methods have primarily focused on tabular data and machine learning. Several methods have been proposed to balance the trade-off between privacy and fairness in classification, both theoretically  and empirically \cite{cummings2019compatibility, pujol2020fair,xu2019achieving}. Of particular relevance to our work, \citet{ganev2022robin} demonstrated the ``Matthew Effects'' (\ie better performance for majority groups) of DP tabular synthetic data, which is the main focus of our experiments. 

\paragraph{Private mean estimation} 
The first part of our paper deals with finding a private empirical mean of a distribution, which is necessary because empirical means have been shown to reveal personally identifiable or otherwise sensitive information \cite{dinur2003revealing,dwork2017exposed,dwork2015robust}. Foundational work by \citet{karwa2017finite} studied algorithms for privately calculating  statistical properties of finite-sample Gaussian distributions in various settings. Follow-up work introduced practical methods for incorporating distributional assumptions for multivariate settings \cite{biswas2020coinpress,kamath2020private} or for long-tailed distributions \cite{kamath2022private}.
Complementary work discussed robustness guarantees to data ablations \cite{liu2021robust}, and eschewing distributional assumptions on Gaussians \cite{ashtiani2022private}. 
Of greatest relevance to our work is that of \citet{biswas2020coinpress}, who present an adaptive mean estimation algorithm, which we discuss in detail later. Additionally, contemporaneous work by \citet{lin2023differentially} studied DP stratification for confidence intervals, but they do not consider distributional assumptions on strata group sizes as we do.

\paragraph{Private synthetic data}
A substantial amount of work on DP data synthesis has been conducted in recent years \cite{aydore2021differentially,boedihardjo2022private,cai2021data,mckenna2019graphical,rosenblatt2020differentially,vietri2020new,zhang2021privsyn}, with the best-performing methods following the ``\textit{Select, Measure, Project}'' paradigm (discussed further in our results section) \cite{tao2021benchmarking,rosenblatt2022epistemic}. We present our synthetic data results for three state-of-the-art algorithms: MST \cite{mckenna2021winning}, GEM \cite{liu2021iterative} and AIM \cite{mckenna2022aim}, each of which offers a different flavor of this paradigm. 

\section{Preliminaries and Problem Statement}
\label{sec:prelims}

\paragraph{Notation} 
We notate a standard Gaussian normal distribution 
as ${\mathcal {N}}(\mu ,\sigma ^{2})$, using $\mu$ for mean and $\sigma^2$ for variance. 
A common distribution in data privacy is the centered \textit{Laplace} distribution, notated $Lap(b)$, with PDF $f(x| b) = {\frac {1}{2b}}\exp \left(-{\frac {|x|}{b}}\right)$. 
We use $X$ to denote the data set in consideration, with $x_1,...,x_n \in X$ items corresponding to individuals. Each $x_i$ is a single value or a vector of multiple values. We let $G_1 \cup \ldots \cup G_k = X$ denote disjoint subsets of $X$.

\paragraph{Differential privacy basics} 
Differential privacy (DP) guarantees that the result of a data analysis or a query remains virtually unchanged even when one record in the dataset is modified or removed, thus preventing any deductions about the inclusion or exclusion of any specific individual. Modifying or removing a record from dataset $X$ induces a \textit{neighboring} dataset $X'$. We usually fix our definition of neighboring datasets, and define DP accordingly. For the purposes of our paper, $X$ and $X'$ are neighboring if one can be obtained by removing a single item $x_i$ from the other. 

The definitions for classical DP and zero-concentrated DP ($\rho$-zCDP) (a common alternative definition relevant to our work) are given in Definition~\ref{def:dp} and Definition~\ref{def:zcdp} respectively.
\begin{definition}[$(\epsilon,\delta)$-Differential Privacy]
\label{def:dp}
    A randomized mechanism ${\mathcal {M}}$ provides $(\epsilon,\delta)$-differential privacy if, for all pairs of neighboring datasets $X$ and $X'$, and all subsets $R$ of possible outputs:
$${\displaystyle \Pr[{\mathcal {M}}(X)\in R]\leq e^{\epsilon} \Pr[{\mathcal {M}}(X')\in R] + \delta}$$
\end{definition}
\begin{definition}[Concentrated Differential Privacy ($\rho$-zCDP)~\cite{bun2016concentrated}]
\label{def:zcdp}
    Here, $D_\alpha\left(\mathcal{M}(X)||\mathcal{M}(X')\right)$ denotes
    $\alpha$-R\'enyi divergence. Then, a randomized algorithm $\mathcal{M}$
    satisfies \emph{$\rho$-zCDP} if for
    all pairs of neighboring datasets $X$ and $X'$,
    $$D_\alpha\left(\mathcal{M}(X)||\mathcal{M}(X')\right) \leq \rho\alpha, \forall \alpha \in (1,\infty)$$
\end{definition}
These two closely related definitions 
scale with different relative privacy parameters. As \citet{bun2016concentrated} showed, an ordering over the guarantees is as follows: An $(\epsilon,0)$-DP mechanism gives $\frac{\epsilon^2}{2}$-zCDP, which gives $(\epsilon \sqrt{2 \log(1/\delta)}, \delta)$-DP for every $\delta > 0$.

\begin{definition}[Sensitivity] Let $f: X \rightarrow \mathbb{R}$ be a real-valued function. The \textit{sensitivity} $\Delta f$ of $f$ is defined as:
    $\Delta f=\max |f(X)-f(X')|$, where the maximum is taken over all pairs of possibly neighboring data sets $X,X'$.
\end{definition}
\begin{definition}[Laplace Mechanism] Given a real-valued function $f$, the Laplace mechanism provides ``pure'' $(\epsilon,0)$-differential privacy:
    ${\mathcal {M}}_{Lap}(f, X,\epsilon)=f(X)+Lap \left(\frac{\Delta f}{\epsilon }\right).$
\label{def:lap}
\end{definition}
To understand the impact of adding Laplace noise on the utility of a DP estimate, we will also require the following standard tail bound for our theoretical analysis: 
\begin{definition}[Laplace Tail Bound]
\label{def:laplace_tail}
Let $Y$ be a random variable draw from a centered Laplace distribution with parameter $b$. Then
\begin{align*}
\Pr[|Y| \geq \alpha b] \leq e^{-\alpha}.
\end{align*}
\end{definition}
\begin{definition}[Gaussian Mechanism ($\rho$-zCDP)] The Gaussian mechanism provides $\rho$-zCDP:
\begin{align}
    {\mathcal {M}}_{Gaussian}(f, X, \rho)=f(X)+ \mathcal{N}\left(0,\frac{\Delta f^2}{2\rho}\right).
\end{align}
\label{def:gauss}
\end{definition}
\begin{definition}[Composition rules \citep{bun2016concentrated,dwork2014algorithmic}] 
\label{def:comp}
$(\epsilon,\delta)$-DP composes gracefully. For ``sequential composition,'' if two randomized algorithms $\mathcal{M}$ and $\mathcal{M'}$ satisfy $(\epsilon_1,\delta_1)$-DP and $(\epsilon_2,\delta_2)$-DP, respectively, then the sequential $M^*(X) = (\mathcal{M}(X), \mathcal{M'}(X))$ satisfies $(\epsilon_1 + \epsilon_2,\delta_1 + \delta_2)$-DP. For ``parallel composition,'' if $\mathcal{M}$ satisfies $(\epsilon,\delta)$-DP, and $\{X_1,...,X_k\}$ are disjoint, then the parallel mechanism $\mathcal{M}^*(X) = (\mathcal{M}(X_1), ..., \mathcal{M}(X_k))$ also satisfies $(\epsilon,\delta)$-DP. Note that $\rho$-zCDP composes analogously.
\label{def:comp}
\end{definition}

\paragraph{Problem statement} 
We consider a setting where the dataset $X$ can be divided into groups of individuals. Specifically, we assume $k$ \emph{disjoint} subsets of $X$: $G_1, \ldots, G_k$, each of size $\{g_1, g_2,...,g_k\} = \mathbf{g}$, that partition $X$ s.t. $G_1 \cup G_2 \cup ... \cup G_k = X$. Groups are typically defined by \emph{sensitive attributes}. For example, suppose we have two sensitive attributes, race with $N_r$ possible values, and gender identity with $N_g$ values. Then we would create a group $G_i$ for each of the $N_r \cdot N_g$ possible combinations of race and gender identity. Our goal will be to release DP statistics about each group $G_i$ so as to prioritize the highest possible accuracy for each group while also achieving acceptable accuracy for the full population when those statistics are aggregated.

\paragraph{Access to public weights} In this paper, we assume limited access to public data: namely, available estimates of group sizes. This data is often already available. For example, in the case of intersectional groups (\eg between race, gender identity, and income), the proportion of these groups in many populations is known (\eg from Census data) across data contexts. Studying the implications of public data access for privacy is common, and assuming information about group sizes is quite mild in comparison to assumptions made in most prior work in this area
\cite{bie2022private,ji2013differential, liu2021leveraging}. Without accurate estimates of group sizes, we expect that the performance of stratification methods in approximating global statistics would degrade, although this topic is beyond the scope of our paper. 

\section{Stratified Private Mean Estimation} \label{sec:strat_mean}
We begin by studying the effect of stratification on the problem of DP mean estimation for (single-variate) Gaussian data. We consider the standard simplified setting where the data variance is fixed and known, so data can be scaled to have variance $1$. \Ie $X = \{x_1, x_2, ..., x_n\}$ consists of $n$ i.i.d. draws from $\mathcal{N}(\mu, 1)$ where $x_i \in \mathbb{R}$. Additionally, we assume a known upper bound $R$ on the absolute value of the mean, $\mu$. In this setting, the standard DP estimator combines the Laplace mechanism with a data clipping step \citep{dwork2009differential, karwa2017finite}. For a scalar value $x$ and a chosen constant $\gamma \in (0,1)$, we let:
\begin{align*}
    \clip(x) = \begin{cases}
    x ~~~~~ \text{if } x\in [R-\sqrt{\log n/\gamma}, R+\sqrt{\log n/\gamma}]\\
    \sign(x) \cdot (R+\sqrt{\log n/\gamma}) ~~~~~ \text{otherwise}.
    \end{cases}
\end{align*}
We then return the empirical mean with clipping and Laplace noise as a DP estimate $\hat{p}$. \Ie
\begin{align*}
\hat{p} = \frac{1}{n} \sum_{i=1}^{n} \clip(x_i) + Lap\left(\frac{2R + 2\sqrt{\log n/\gamma}}{n\epsilon}\right).
\end{align*}
It can be checked that $\hat{p}$ is $(\epsilon,0)$-differentially private. To bound estimate error, we can use a triangle inequality involving the (non-private) empirical mean ${p} = \frac{1}{n}\sum_{i=1}^n x_i$:
\begin{align*}
    |\hat{p} - \mu| \leq |{p} - \mu| + |\hat{p} - {p}|.
\end{align*}
By standard Gaussian concentration, we have that the first term, which represents inherent statistical error, is bounded by $O(\frac{1}{\sqrt{n}})$ with high probability. We'd like to bound the second term, which represents additional error incurred by privatization. Applying Definition~\ref{def:laplace_tail}, this second term is bounded with probability $1-O(\gamma)$ by:
\small
\begin{align}
\label{eq:fresh_est_bound}
    |\hat{p} - {p}| \leq O\left(\log(1/\gamma) \cdot \left(\frac{R + \sqrt{\log n/\gamma}}{n\epsilon}\right)\right).
\end{align}
\normalsize

In the setting where $X$ can be split into disjoint groups, $G_1, \ldots, G_k$, we assume that each group $i$ contains normally distributed data with mean $\mu_i$ and standard deviation $1$. \Ie $X$ follows a mixture of unit-variance Gaussian distributions.\footnote{DP methods for estimating mixtures of Gaussians were studied extensively by \citep{kamath2019differentially}, differing in that the identities of the sub-populations are \textit{unknown}.} Since the global mean can easily wash out information about individual groups, the natural stratification approach to minimizing disparate impact would be to compute DP estimates $\hat{p}_1, \ldots, \hat{p}_k$ for each $\mu_1,\mu_2,...\mu_k$. \Ie
\small
\begin{align}
\label{eq:hatpi}
\hat{p}_i = \frac{1}{|G_i|} \sum_{x\in G_i} \clip(x) + Lap\left(\frac{2R + 2\sqrt{\log |G_i|/\gamma}}{|G_i|\epsilon}\right).
\end{align}
\normalsize

Since each of these means is based on a disjoint subset of $X$, by parallel composition (Definition \ref{def:comp}), we can compute group-specific means with privacy parameters $(\epsilon,0)$, and still obtain an overall $(\epsilon,0)$-private method. 

Given individual group estimates, how do we then compute an estimate of the global mean, $\mu$? One natural approach is to do so from scratch, using a different DP estimator. Doing so, however, has a few drawbacks: (1) Since $X$ is not disjoint from $G_1, \ldots, G_k$, we would only obtain a $(2\epsilon,0)$-private method via serial composition if we report all individual means as well as the group mean. That is, for the same level of privacy, our error will be greater by a factor of two. (2) If we separately estimate the global mean, then this estimate may be ``inconsistent,'' in that it could differ from a weighted average of the per-group means. An alternative is to simply return the private estimate:
\small
\begin{align}
\label{eq:p_strat}
\hat{p}_{strat} = \frac{1}{n} \sum_{i=1}^k |G_i| \cdot \hat{p}_i. 
\end{align}
\normalsize
\begin{proposition}[Privacy of $\hat{p}_{strat}$] \label{prop:privacy}The stratified mean with the Laplace mechanism is $(\epsilon, 0)$-DP by parallel composition rules (see Definition~\ref{def:comp}).
\end{proposition}
But how accurate will the estimate be in comparison to a ``fresh'' DP estimate of the global mean? Again, letting ${p}$ equal the empirical mean $\frac{1}{n}\sum_{i=1}^n x_i$, we find the following (complete proof deferred to the Appendix):
\begin{proposition}[Worst Case Bound for Stratified Mean Estimation]
Let $\hat{p}_{strat}$ be the estimator defined in Eq.~\eqref{eq:p_strat}. With probability $1-O(\gamma)$, 
\label{prop:naive}
\small
\begin{align}
\label{eq:pstrat_error}
    |{p}-\hat{p}_{strat}| \leq O\left(\log(1/\gamma) \cdot \sqrt{k}  \left(\frac{R + \sqrt{\log(n/k\gamma)}}{n\epsilon}\right)\right)
\end{align}
\normalsize
\end{proposition}
When $k$ is small (\eg $k = n^c$ for constant $c < 1$), then $O(\log(n/k\gamma) = O(\log(n/\gamma))$. So, the above bound appears \emph{worse} than what we obtain from the standard DP estimate $\hat{p}$ in Eq.~\eqref{eq:fresh_est_bound} by a multiplicative factor of $\sqrt{k}$. Nevertheless, when we compute $\hat{p}_{strat}$ in practice, \textbf{we find its accuracy is competitive with $\hat{p}$.} It is natural to ask why this is the case. 

\paragraph{Distributional assumptions} 
To better understand this question, first note that the error term involving $R$ is typically a lower-order term, since an accurate estimate for $R$ can be found using adaptive DP mean estimation techniques \cite{biswas2020coinpress}. Removing all shared multiplicative terms and assuming $\gamma$ is a small constant, we then see that the difference between the error of a fresh DP estimate, as in Eq.~\eqref{eq:fresh_est_bound}, and $\hat{p}_{strat}$, as in Eq.~\eqref{eq:pstrat_error} is a matter of $\log(n)$ vs. $k\log(n/k)$. In the proof of Proposition \ref{prop:naive}, the $k\log(n/k)$ term arises from the following sum involving the group sizes:
$\sum_{i=1}^k \log(|G_i|).$

\begin{figure}
\centering
\fbox{\includegraphics[width=.3\textwidth]{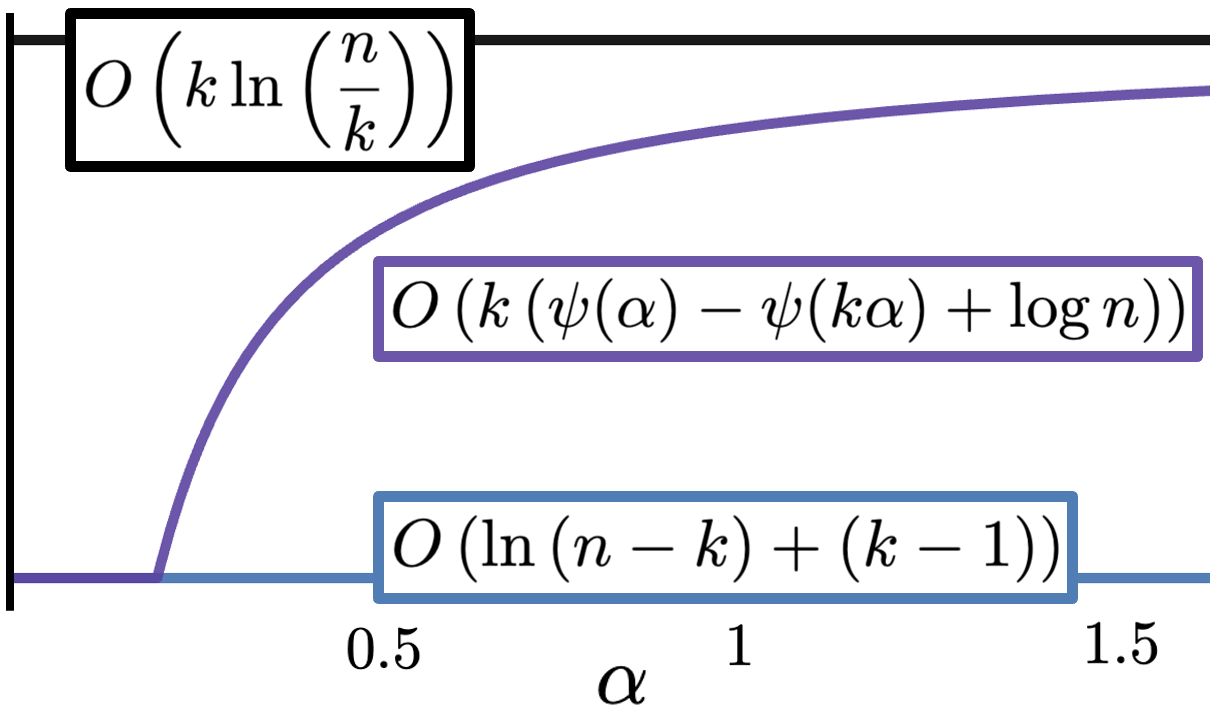}}
\caption{Visualizing the affect of the $\alpha$ parameter on this expectation term, with fixed $k$ and $n$. $\alpha$ essentially controls the ``sparsity'' of the distribution; a small $\alpha$ implies expected error closer to the theoretical lower bound. }
\label{fig:alpha_effect}
\end{figure}

This sum is maximized when all sizes are equal, so $|G_i| = n/k$. However, when the sum is smaller, we obtain a tighter bound than Proposition \ref{prop:naive}. In the most extreme case, when all group sizes equal one, except for a single (majority) group, the bound can be as good as $(k-1) + \log(n-k)$, which nearly matches the $\log n$ dependence of $\hat{p}$. And, in fact, we \textit{rarely} observe uniform group sizes in practice, particularly when considering intersectional groups. Often, a small number of majority groups dominate. To better explain the error, take a standard Bayesian assumption that our group size vector is drawn from a \textit{Dirichlet} distribution, $\mathcal{D}(\alpha, k)$. (Our Appendix contains notes and details on Dirichlet parameters and behavior for completeness.) In this case, we prove the following bound (see Appendix):
\begin{theorem}[Error Upper Bound with Dirichlet Assumption] Consider the Dirichlet distribution $\mathcal{D}(\alpha, k)$ with parameters $k,\alpha$ and let $\mathbf{g} \sim \mathcal{D}(\alpha, k)$ be a vector drawn from this distribution. If for each group $i$, $|G_i| = g_i/n$, then,  letting $\psi$ denote the digamma function, we have: 
\label{thm:dirichlet_expression}
\begin{align}
\E\left[\sum_{i=1}^k \log |G_i|\right] \leq k\left(\psi(\alpha)-\psi(k\alpha) + \log n\right) 
\end{align}

\end{theorem}

To better understand the bound of Theorem \ref{thm:dirichlet_expression}, please refer to Figure~\ref{fig:alpha_effect}, which plots the bound in comparison to the weaker upper bound of $k\log(n/k)$ from Propostion \ref{prop:naive}, and in comparison to  our informal lower bound $O(k + \log(n))$. As we can see, as the parameter $\alpha$ of the Dirichlet-distributed group size vectors varies from $[0.2, 1.0)$, 
we interpolate between the lower and upper bound. Small $\alpha < 1.0$ implies ``sparsity'', i.e. one or two dominant entries in the vector $\mathbf{g}$. 
Overall, we frame this result as follows: In cases when minority groups in the data are relatively small compared to the majority groups, additional noise from privacy is expected to be small when aggregated.

\paragraph{Practical Implications of Theory} Our theory for stratified private mean estimation (PME) bounds worst case costs (Proposition~\ref{prop:naive}) and expected costs (Theorem~\ref{thm:dirichlet_expression}) for stratification. Much can be learned from these bounds practically. For a salient example, consider a sample of American Community Survey data for New York from 2018, where $n=196967$ and the \texttt{RACE} variable has $k=9$. The group size vector here is $v_g= [0.70, 0.12, 0.086, 0.056, 0.029,...]$. Using an iterative Monte Carlo approach to find a likely Dirichlet prior over $v_g$, we find $\alpha\approx 0.13$. Applying Proposition~\ref{prop:naive} and Theorem~\ref{thm:dirichlet_expression} then tells us that, for any setting of $R$ or $\epsilon$, our \textit{expected} absolute error from stratification in estimating the mean for a given variable is approximately $\sim 80\%$ better than the worst case error, relative to the best case stratification error. This at least partially explains the strong performance of stratification empirically, seen in Figure 2. Put another way, our theory shows that, when the group sizes follow a Dirichlet distribution, then, as $\alpha \rightarrow 0$, the error of stratified PME scales with $O(k+ \log(n) + R)$.  This is close to the error of non-stratified PME, which scales as $O(\log(n) + R)$. 

\section{Adaptive Private Mean Estimation} 

\label{sec:adaptive_algos}
In the previous section, we focused on bounding the additional error introduced by differential privacy, \ie on $|\hat{p} - p|$, where $\hat{p}$ is a DP mean estimate and $p$ is the empirical mean. However, as we highlighted, this error is always in addition to the inherent statistic error $|\mu - p|$ between the true population mean $\mu$ and our empirical estimate $p$. It is well known that stratification can help reduce this statistic error \cite{cochran1977sampling, BotevRidder:2017}. This helps explain the value of stratification in practice: extra error introduced by DP noise can be offset by reduced statistical error. 

In particular, consider the case when each group has fixed variance $1$, but the means $\mu_i$ for the groups can differ substantially i.e. the overall data variance $\sigma$ is much larger than $1$. Then we expect statistical error on the order of $O(\sigma/\sqrt{n})$. On the other hand, if we assume we know exact group sizes and stratify, statistical error should scale as $O(1/\sqrt{n})$, which can be substantially better \cite{BotevRidder:2017}. 

The naive mean estimation method analyzed in the previous section will likely not benefit from this scenario, since we will need to choose a large range $R$ (thus scaling our noise) if our group means differ substantially. However, adaptive mean estimation methods like \textsc{Coinpress} \cite{biswas2020coinpress} have been introduced that only depend logarithmically on $R$. In our experiments, we observe that the aggregate population level means of these methods can actually \textit{benefit} from stratification. 

\begin{figure}[t!]
    \centering
    \includegraphics[width=0.49\textwidth]{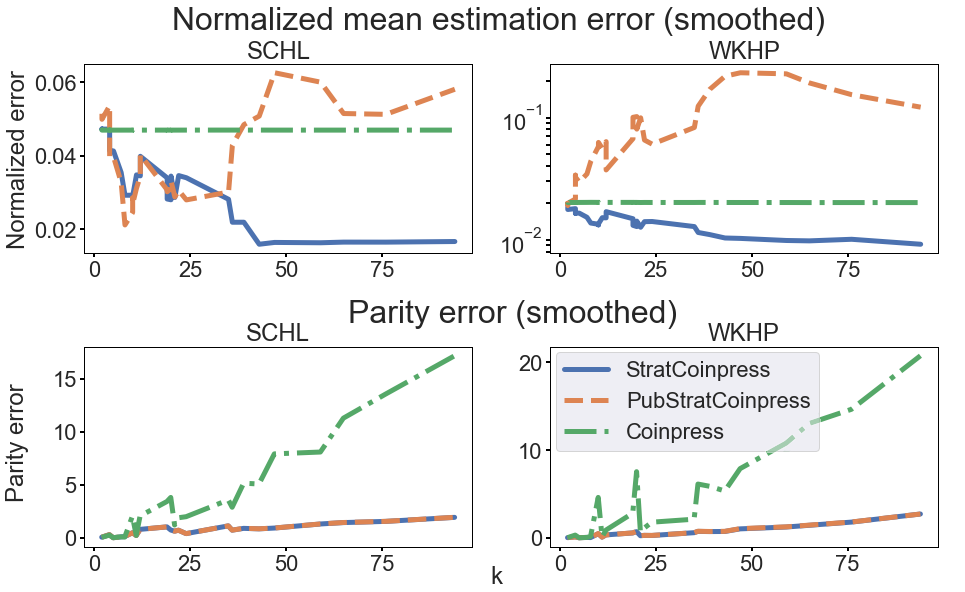}
    \caption{Comparisons of \textsc{Coinpress} and \textsc{StratCoinpress} on U.S. Census data from Folktables; $\texttt{SCHL}$ is years spent in school (ordinal, 0-24), $\texttt{WKHP}$ is hours worked per week (ordinal, 0-168). Top row shows normalized mean estimation error, bottom row shows parity error, both as the number of groups $k$ increases.}
    \label{fig:folktables_coinpress_results}
\end{figure}

\paragraph{Measuring disparate impact of DP} 
We consider the following formalism for our setting to measure the impact of private mechanisms on subgroups in data. Intuitively, we want the error for any protected group in a population to be comparable to the error for other groups. Consider dataset $X$, which consists of $k$ strata $G_1 \cup ... \cup G_k$. Consider a non-private function that maps $X$ to a real-valued vector with values for each of the $k$ strata as well as a single global value $f(X): X \rightarrow \mathbb{R}^{k+1}$. A privatized version of $f(x)$ is denoted $\mathcal{M}_f(X): X \rightarrow \mathbb{R}^{k+1}$. Results for both on stratum $i$ are denoted $f(G_i)$ and $\mathcal{M}_f(G_i)$ respectively. Population level results are then $f(X)_{k+1}$ and $\mathcal{M}_f(X)_{k+1}$.
\begin{definition}[Parity error]
    Parity error $\phi$ is the average normalized absolute error in approximating $f(G_i)$ for each stratum $i$, plus the normalized absolute error in approximating $f(X)_{k+1}$, weighted by a positive parameter $\omega$ that determines the emphasis on population-level accuracy:
    \small
    \begin{align}
        \omega \left| \frac{f(X)_{k+1} - \mathcal{M}_f(X)_{k+1}}{f(X)_{k+1}}\right| + \sum_{i=1}^k \left|\frac{f(G_i) - \mathcal{M}_f(G_i)}{f(G_i)}\right|
    \end{align}
    \normalsize
    \label{def:di}
\end{definition}
In Definition~\ref{def:di}, $\omega$ weights the faithfulness to the overall population estimate in the sum. We emphasize faithfulness for the protected class strata, and so we set $\omega = \frac{1}{k}$ (weighting population-level estimates the same as any single-stratum estimate).

\paragraph{Data and Setup} We provide results of two experiments to validate the stratified adaptive private mean estimation approach.
The first experiment (Figures~\ref{fig:error_vs_n},~\ref{fig:error_vs_k},~\ref{fig:exhaustive}) uses a synthetic mixture of $k$ Gaussians, over a dataset of size $n$, with Dirichlet parameter $\alpha$, to compare \textsc{Coinpress} and \textsc{StratCoinpress}. (The Appendix includes parameters for our synthetic Gaussians, which are chosen illustratively and are independent of our theory). Our second experiment (Figure~\ref{fig:folktables_coinpress_results}) is on demographic census data for New York State from \textit{Folktables}, a standard dataset in fair-ML literature \citep{ding2021retiring}. For consistency and comparability with prior work, we maintain Census variable codes in all of our plots and results; for example, \texttt{SCHL} is a discrete variable denoting years of education, and ranges from 0-24. For a complete list of Census variables, counts, marginals, and their corresponding meanings and domains, see the Appendix. We also use common, legally protected classes, like \texttt{SEX}, \texttt{AGE} and \texttt{RAC1P}, to create groups for stratification. Note that we present two stratified variants: \textsc{StratCoinpress} assumes direct access to the demographic weights necessary for stratification, while \textsc{PubStratCoinpress} calculates those weights on a smaller, non-overlapping public holdout set.

\begin{figure}[ht]
\centering
\includegraphics[width=0.47\textwidth]{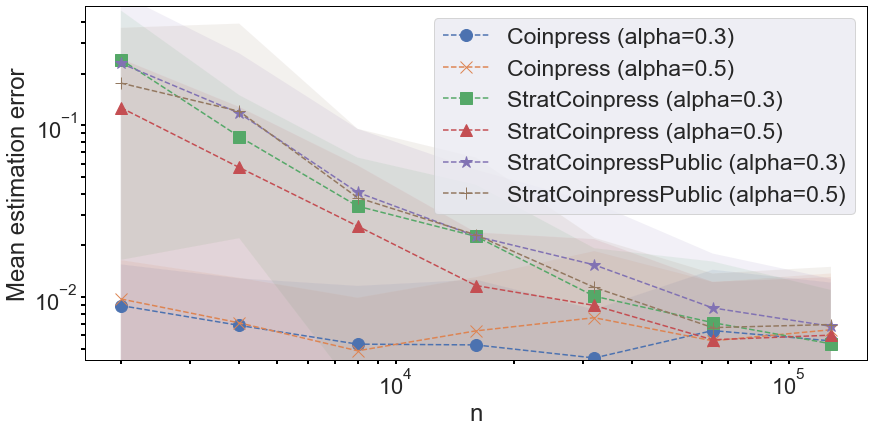}
\caption{Mean estimation error, varying data size $n$ in synth Gaussian mixture (50 runs). Stratified variants converge to performance of non-stratified \textsc{Coinpress} as $n$ grows (for $n \geq 10000$, error $\leq 1\%$). In other words, with large samples, we don't have to ``pay'' in \textit{global} estimation error for a reduction in \textit{group-specific} error.}
\label{fig:error_vs_n}
\end{figure}

In Figure~\ref{fig:error_vs_n}, we depict the convergence of \textsc{StratCoinpress} to the performance of  \textsc{Coinpress} as $n$ grows linearly. 
We do not report on disparate impact experiments with the synthetic data, as we can arbitrarily \textit{increase} harm (measured by parity error) by increasing variance of the mixed Gaussian means. Figures~\ref{fig:exhaustive} and~\ref{fig:error_vs_k} on this experiment, which show the effects of varying $k$ and $\alpha$, are available in the Appendix. 

\begin{figure*}[ht!]
\centering
\includegraphics[width=\textwidth]{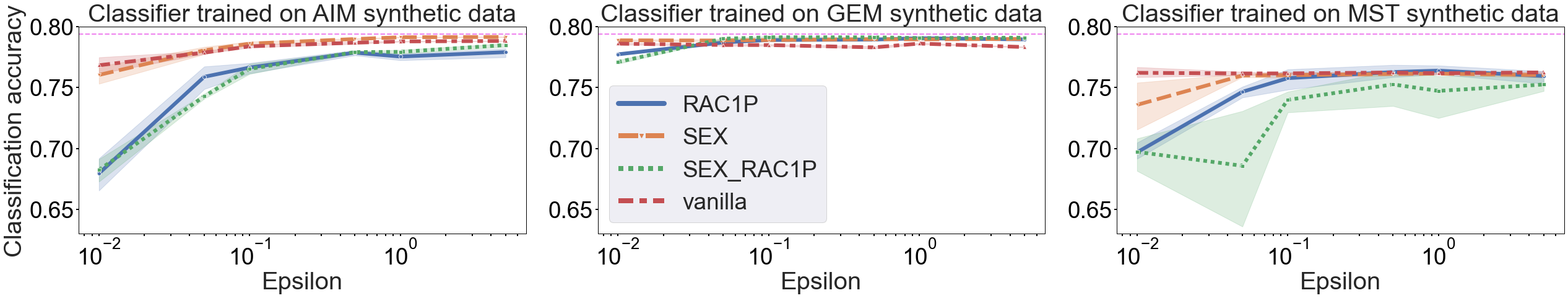}
\caption{Performance of different stratified synthesizers, Folktables employment prediction. Note that variants of GEM and AIM approach a classifier trained on real data (accuracy $\approx 0.8$) as privacy budget increases. From the perspective of maintaining predictive utility, stratification appears to incur only minor costs, if any (for $\epsilon \geq 1$, difference $\leq 2\%$).}
\label{fig:performance_acs}
\end{figure*}

\paragraph{Effectiveness}
In Figure~\ref{fig:folktables_coinpress_results}, we show the effectiveness of a stratified adaptive mean estimation in practice. We vary $k$ by creating a different number of intersectional groups by combining $sex \in [0,1]$, $race \in [0,8]$, $age \in [0,4]$ (bucketed), and $physical$ $disability \in [0,1]$. The top row of Figure~\ref{fig:folktables_coinpress_results} shows that the error of the stratified variants of $\textsc{Coinpress}$ is controlled in settings where the range of possible values is small relative to the total population size. In fact, when the subgroups are particularly meaningful (such as with $\texttt{SCHL}$, $\texttt{WKHP}$ and $\texttt{JWMNP}$), stratification can sometimes \textit{improve} on the overall mean estimates. Challenging settings for stratification are large continuous spaces such as with $\texttt{PINCP}$, where the strata-specific estimates (and, thus, the accuracy of aggregation) suffer from well-known private mean estimation challenges for wide range long-tailed distributions.

The bottom row of Figure~\ref{fig:folktables_coinpress_results} tells the story of harm reduction on subgroups in the data. Because private estimates are done over disjoint strata under composition (Def.~\ref{def:comp}), parity error can be tightly controlled by $\textsc{StratCoinpress}$.
On the other hand, non-stratified $\textsc{Coinpress}$ incurs substantial utility loss for most intersectional groups, because their mean differs significantly from the population-level mean.

\begin{figure}[!ht]
\centering
\includegraphics[width=0.45\textwidth]{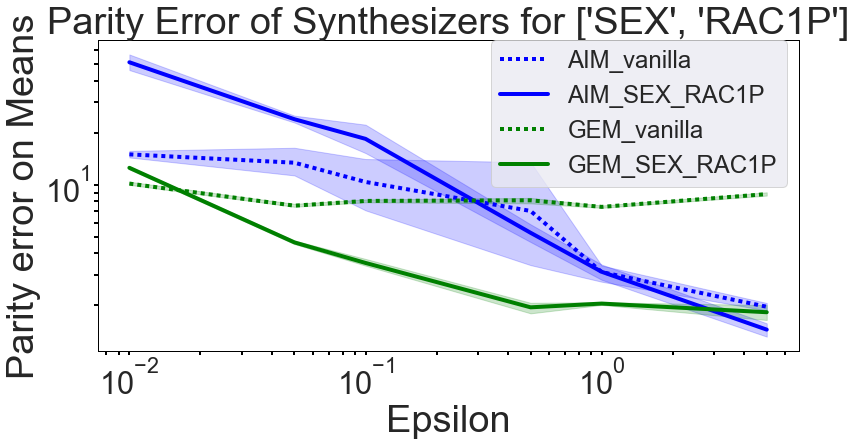}
\caption{Parity error, where $f$ aggregates the error of means calculated on all variables in the Folktables employment prediction dataset. Here we contrast the two higher performing synthesizers: AIM (adaptive, lower predictive performance) and GEM (non-adaptive, higher predictive performance). Vanilla GEM struggles without explicit stratification, but by stratifying we can reduce parity error significantly (a reduction of $\geq 200\%$). Vanilla AIM's budget and workload adaptivity has parity error benefits at low privacy budgets, though explicit stratification improves at higher budgets.}
\label{fig:parity_error}
\end{figure}

\section{Stratified Data Synthesis}

\label{sec:synthetic}
The disparate impacts of DP have been most commonly framed as occurring through private synthetic data \citep{bagdasaryan2019differential,ganev2022robin}. Luckily, the principles of stratification (subset, estimate and aggregate) are not limited to mean estimation. Stratifying private synthesizers follows the same straightforward process: We learn separate parametric private distributions for each group-specific stratum. Then, to compose population-level data, we sample elements from strata-specific models \textit{proportionally to their representation in the full population}. We do not offer theoretical bounds for private synthesizers as we did for mean estimation, as their guarantees are limited to specific query workloads selected during model fitting. This makes theoretical analysis challenging, and we leave it to future work. 

\paragraph{Synthesizers and Privacy Settings} 
We conducted an extensive empirical evaluation on the viability of stratified synthetic data. We selected three state-of-the-art DP synthesizers: MST \citep{mckenna2021winning}, GEM \citep{liu2021iterative}, and AIM \citep{mckenna2022aim}. 
MST is data-aware, GEM is both data- and workload-aware, and AIM is data-, workload-, and privacy budget-aware \cite{mckenna2022aim}. Benchmarking DP synthesizers is non-trivial and computationally expensive \citep{rosenblatt2022epistemic}. We ran our experiments on an NVIDIA T4 GPU cluster and on a high-performance CPU cluster. We ran our models on the same privacy settings as \cite{mckenna2022aim}, $\epsilon \in \{0.01,0.05,0.1,0.5,1.0,5.0\}$, representing a low to medium privacy budget regime. We trained 5 differently seeded models for each synthesizer at each privacy setting. Our models took over 200 hours of compute time to fit.
We employ two naming conventions in labeling our figures: the first is that $\_$\texttt{vanilla} denotes the non-stratified version of the algorithm, and the second is that $\_$\texttt{VARIABLE} means that the algorithm stratifies explicitly along a set of dimensions (for example, \texttt{SEX}$\_$\texttt{RAC1P} implies separate models maintained for all subgroups induced by sex and race.)

\paragraph{Classification Setting} First, 
we found that stratification only seemed to harm the overall utility of private synthetic data in low privacy regimes ($\epsilon < 0.1$). We demonstrate this on the Folktables employment prediction task \cite{ding2021retiring} by training a classifier on DP synthetic versions of the data. Figure~\ref{fig:performance_acs} shows that the strength of these classifiers increased as $\epsilon$ increased, implying that relationships in the data are maintained for nearly all variants, although we acknowledge the limitations of using classification as a proxy task higher-dimensional fidelity.

\paragraph{Parity Error Reduction} Second, we find that stratified variants of all synthesizers reduce parity error (Definition~\ref{def:di}). Figure~\ref{fig:parity_error} shows how additional privacy budget can generally improve the performance of the stratified variants (where the parity error function $f$ is the aggregate normalized difference of means across all variables in the Folktables employment task data). The best-performing synthesizer in our tests (GEM) also stratifies most gracefully, and provides the best parity error in nearly  all $\epsilon$ parameter settings. We defer some experimental results to the Appendix, which further demonstrate that, for all synthesizers, as privacy budget increases: (1) the maximum disparity of all \textit{means} gradually decreases (Figure~\ref{fig:max_disp}); (2) the \textit{demographic parity} \cite{hardt2016equality} on the Folktables employment task increases towards 1 (Figure~\ref{fig:dem_par_ratio}); and (3) the maxmimum false negative rate  difference between groups decreases (Figure~\ref{fig:fnr_difference}).

Finally, in Figure~\ref{fig:marginals} (also deferred to the Appendix), we use the \textsc{All 3-Way Marginal} workload error from \cite{mckenna2022aim} as our function for parity error. This is a standard method of measuring performance of DP synthesizers: a $k$-way marginal for attribute set $S$ (where $|S|=k$) is a histogram over $x_S \in X$; marginal workload error is the average difference between these marginals run on real data and on private synthetic data (see Appendix or \cite{mckenna2022aim} for more details). Figure~\ref{fig:marginals} demonstrates that, in the case of non-budget adaptive algorithms like GEM, one can stratify and achieve low parity error and great overall performance. In the case of the budget-adaptive algorithm AIM, we see that it eventually has ample budget to adapt to all groups, and that the stratification might hurt performance. However, \textit{we rarely know what an acceptable budget setting is a priori}; insufficient privacy budget for AIM (here, $\epsilon < 1.0$) greatly increases the parity error of the \texttt{vanilla} adaptive algorithm relative to its stratified variant. Performance of the the stratified variant improves stably for all subgroups, and is thus much safer to use in regimes with a limited privacy budget, or when it is unclear how to set the budget.

Overall, we found that the stratified GEM variants convincingly outperformed MST and AIM with small privacy budget ($\epsilon\leq 1.0$). 
AIM did outperform GEM in the highest privacy regime of $\epsilon=5.0$, likely due to AIM's ability to adapt and utilize ``excess'' privacy budget. 

\paragraph{Limitations and future work} Our experimental results suggest that stratified synthetic data is safe given sufficient privacy budget, and that it often helps improve utility of subgroup data while preserving population-level utility, even in the case of adaptive algorithms. However, this fundamentally relies on good aggregation proportions, which may not always be available (say, in a medical context for a \textit{specific} hospital). Though we believe that stratification provides a strong baseline for future work on adaptive DP algorithms with disparate impact protections, a stronger theoretical underpinning for stratified private synthesizers would allow for greater confidence in deployment. 
Additionally, our approach essentially targets parity error, and perhaps unsurprisingly trades off good performance there with worse performance by other metrics. Future work could formally characterize this trade-off by identifying the Pareto frontier of this dual optimization problem.

\section{Conclusion}

Reducing disparate impact in DP data release is a laudable goal that has received much attention in recent years. We showed that, when access to public estimates of group proportions in the data can be assumed, a \textit{stratified} approach to disparate impact reduction is surprisingly effective, and that it does not significantly reduce --- and sometimes even improves --- the accuracy of population-level private statistics when the stratified private data is aggregated. With this work, we hope to encourage interest in principled methodologies for harm reduction when privatizing social data.  We also hope that that practitioners will find the simple strategy we outlined here immediately applicable for private data release on sensitive data with protected classes. 

\input{ack}

\bibliography{main}

\newpage
\quad
\newpage
\appendix

\section{Complete proofs}
\label{sec:appendix:proofs}

Below we provide complete proofs and results to support the main body of the paper. 

\subsection{Proof of Proposition~\ref{prop:naive}}
The standard bound for a private mean estimator with the Laplace mechanism is given in the main paper body. We restate it here for convenience. With probability $1-O(\gamma)$, where $p$ is the empirical mean and $\hat{p}$ is the private empirical mean, then:
\begin{align}
     |\hat{p} - {p}| \leq O\left(\log(1/\gamma) \cdot \left(\frac{R + \sqrt{\log n/\gamma}}{n\epsilon}\right)\right) \label{eq:standard}
\end{align}

Recall that our stratified synthesizer is defined as:
\begin{align}
\label{eq:p_strat}
\hat{p}_{strat} = \frac{1}{n} \sum_{i=1}^k |G_i| \cdot \hat{p}_i \end{align}
Where each $\hat{p}_i$ is a standard private empirical mean obtained through the Laplace mechanism and clipping (see Section~\ref{sec:strat_mean} for details), with the bound given in Equation~\ref{eq:standard}.

In Proposition~\ref{prop:naive}, we bound the error between the stratified private empirical mean estimator $\hat{p}_{strat}$ and the non-private empirical mean $p$. We restate Proposition~\ref{prop:naive} below for convenience.

\paragraph{Proposition~\ref{prop:naive}}[Worst Case Error Bound for Stratified Mean Estimation]
\textit{Let $\hat{p}_{strat}$ be the estimator defined in Eq.~\eqref{eq:p_strat}. With probability $1-O(\gamma k)$, }
\begin{align}
\label{eq:pstrat_error}
    |{p}-\hat{p}_{strat}| \leq O\left(\log(1/\gamma) \cdot \sqrt{k} \cdot \left(\frac{R + \sqrt{\log(n/k\gamma)}}{n\epsilon}\right)\right)
\end{align}

\begin{proof}
Consider a single group $|G_i|$. When computing the standard DP estimator $\hat{p}_i$ we perform a clip operation for any $x\in G_i$ only if $|x| \geq R + \sqrt{\log|G_i|/\gamma}$. Recall that $x$ is assumed to be Gaussian distributed with mean $\mu_i$ and variance $1$. As is standard in the analysis of the unstratified estimator, we claim that, with high probability, no values of $x$ get clipped. In particular, we apply the standard concentration bound for univariate Gaussians:
\begin{align*}
\Pr[|x-\E[x]| \geq t] \leq e^{-t^2/2}.
\end{align*}
Since $\mu_i$ is assumed to lie in $[-R,R]$, this bound implies that $|x| \geq R + \sqrt{\log|G_i|/\gamma}$ with probability less than $O(\gamma/|G_i|)$. So, by a union bound, no $x\in G_i$ gets clipped with probability at least $1- O(\gamma)$. Further union bounding across all $k$ groups, conclude that with probability at least $1- O(k\gamma)$, no $x$ gets clipped across the entire dataset. This fact makes it easy to write down an expression for ${p}-\hat{p}_{strat}$. In particular, we have:
\begin{align*}
{p}-\hat{p}_{strat} = \frac{1}{n} (L_1 + L_2 + \ldots + L_k),
\end{align*}
where $L_i$ is a centered Laplace random variable with parameter $b_i = \frac{2R + 2\sqrt{\log |G_i|/\gamma}}{\epsilon}$. 
Following the standard definition of sub-exponential random variables found e.g in \cite{wainwright_2019}, we can check that $L_i$ is sub-exponential with parameters $\nu_i = \alpha_i = 2b_i$. 
Accordingly, by a standard argument, $\sum_{i=1}^k L_i$ is sub-exponential with parameters $\nu = 2\sqrt{\sum_{i=1}^k b_i^2}$ and $\alpha = \max_i 2b_i \leq 2\sqrt{\sum_{i=1}^k b_i^2}$ \cite{duchi2021lecture}. We conclude via a standard sub-exponential tail bound (see e.g. Proposition 2.9 in \cite{wainwright_2019}) that for a fixed constant $c$,
\begin{align}
\label{eq:second_to_last}
    \Pr\left[|{p}-\hat{p}_{strat}| \geq c\log(1/\gamma)\cdot \sqrt{\sum_{i=1}^k b_i^2}\right] \leq \gamma.
\end{align}
So, we are just left to bound $\sqrt{\sum_{i=1}^k b_i^2}$. To do so, note that:
\begin{align}
\label{eq:expanded}
\sum_{i=1}^k b_i^2 &= \frac{4}{\epsilon^2} \sum_{i=1}^k (R + \sqrt{\log |G_i|/\gamma})^2 \nonumber\\
&\leq \frac{8}{\epsilon^2} \sum_{i=1}^k R^2 + \log |G_i|/\gamma \nonumber\\
&= \frac{8}{\epsilon^2}\left(kR^2 + k\log(1/\gamma) +  \sum_{i=1}^k\log |G_i|\right)
\end{align}

In Lemma~\ref{lemma:1}, we show that the $\max_{|G_1|,\ldots, |G_k|} \sum_{i=1}^k \log|G_i| \leq k\log(n/k)$ through an arithmetic mean argument, allowing us to upper bound Equation~\eqref{eq:expanded} by:
\begin{align}
    O\left(\frac{k}\epsilon^2 \cdot \left(R^2 + \log(n/k\gamma)\right) \right)
\end{align}
Plugging into \eqref{eq:second_to_last} and using that for any positive $C,D$, $\sqrt{C^2 + D^2} \leq C + D$ yields the proposition.
\end{proof}

Note: for clarity in Lemma~\ref{lemma:1}, we use $|G_i| = g_i$.
\begin{lemma}
\label{lemma:1}
For any $g_1,...,g_k \in \mathbb{R}$,  $g_i > 0$ such that $\sum_{i=1}^k g_i = n$,  we have $\sum_{i=1}^k {\log g_i} \leq k{\log(n/k)}$.
\end{lemma}
\begin{proof}

Note that:
\begin{align}
    \sum_{i=1}^k \log g_i = \log g_1 + ... + \log g_k = \log \left(\prod_{i=1}^k g_i\right)  
\end{align}
Consider any pair $g_i$ and $g_j$, where $g_i \neq g_j$, in $\prod_{i=1}^k g_i$. Let their arithmetic average be $\hat{g}_{ij}$ such that  $g_i=\hat{g}_{ij} + d$ and $g_j=\hat{g}_{ij} - d$, where $d\neq 0$. Then
\begin{align}
     g_i + g_j = \hat{g}_{ij} + \hat{g}_{ij} \\
     g_i \cdot g_j = \hat{g}_{ij}^2 - d^2 < \hat{g}_{ij} \cdot \hat{g}_{ij}
\end{align}
Thus, we can maintain the overall sum $\sum_{i=1}^k g_i = n$ while increasing the product $\prod_{i=1}^k g_i$ by setting $g_i,g_j = \hat{g}_{ij}$ = $n/k$. Accordingly, the product can only be maximized when $g_i = g_j$ for all $i,j$ -- in other words, when $g_1, \ldots, g_k = n/k$. It follows that: 
\begin{align*}
    \log\left(\prod_{i=1}^k g_i\right) &\leq \log\left(\prod_{i=1}^k n/k\right) \\
    \sum_{i=1}^k \log g_i &\leq \sum_{i=1}^k \log n/k.\qedhere
    \end{align*}
\end{proof}

\subsection{Proof of Theorem~\ref{thm:dirichlet_expression}}

We restate Theorem~\ref{thm:dirichlet_expression} here for convenience. 

\textit{Consider the Dirichlet distribution $\mathcal{D}(\alpha, k)$ with parameters $k,\alpha$ and let $\mathbf{g} \sim \mathcal{D}(\alpha, k)$ be a vector drawn from this distribution. If for each group $i$, $|G_i| = g_i \cdot n$, then,  letting $\psi$ denote the digamma function, we have:} 
\begin{align}
\E\left[\sum_{i=1}^k \log |G_i|\right] \leq k\left(\psi(\alpha)-\psi(k\alpha) + \log n\right) 
\end{align}

\begin{proof}
We begin by applying linearity of expectation. 
\begin{align}
\label{eq:plug_back}
\mathbb{E}\left[\sum^k_{i=1} \log |G_i|\right] &= k \mathbb{E}[\log g_i +\log n] \noindent\\
&= k\log n + k \mathbb{E}[\log g_i]
\end{align}

Thus, the main challenge is deriving an expression for $\mathbb{E}[\log g_i]$ where $\mathbf{g} \sim \mathcal{D}(\alpha, k)$. 
We begin by applying law of the unconscious statistician.
\begin{align}
    \mathbb{E}[\log g_i] = \int_0^1 \log g_i \mathcal{D}(g_i | \alpha) dg_i
\end{align}
We can marginalize $\mathcal{D}(g | \alpha)$ out into Beta distribution $\mathcal{B}(\alpha_i, - \alpha_i + \sum^k \alpha_i )$. As we assume symmetric $\alpha$, this becomes $\mathcal{B}(\alpha, k\alpha - \alpha)$. With substitution now, we have:
\begin{align}
    \mathbb{E}[\log g_i ] &= \int_0^1 \log g_i \mathcal{B}(\alpha, k\alpha - \alpha) dg_i \\
    &= \int_0^1 \log g_i \frac{1}{\beta(\alpha, k\alpha - \alpha)} g_i^{\alpha - 1}(1- g_i)^{k\alpha - \alpha - 1} dg_i 
\end{align}
Thus, we will take the square root of our final term.  Note further that $\beta(z_1, z_2) = \frac{\Gamma(z_1)\Gamma(z_2)}{\Gamma(z_1+z_2)} = \int_0^1 t^{z_1-1}(1-t)^{z_2-1}dt$ is the $Beta$ function and $\Gamma(n) = (n-1)!$ is the $Gamma$ function. As both the $Beta$ function term and $\alpha$ are a constants, we can simplify further with:
\begin{align}
    &=\frac{1}{\beta(\alpha, k\alpha - \alpha)} \frac{d}{d\alpha} \int_0^1  g_i^{\alpha - 1}(1- g_i)^{k\alpha - \alpha - 1} dg_i
\end{align}
Note that the integration term is now exactly a $Beta$ function, so we can simplify further.
\begin{align}
    &=\frac{1}{\beta(\alpha, k\alpha - \alpha)} \frac{d \beta(\alpha, k\alpha - \alpha)}{d\alpha}
\end{align}
Which is now equivalent to taking the derivative of $\log(\beta(\alpha, k\alpha - \alpha))$ with respect to $\alpha$.
\begin{align}
    &=\frac{d}{\alpha} \log(\beta(\alpha, k\alpha - \alpha))
\end{align}
Performing the $Gamma$ function expansion we get:
\begin{align}
    &=\frac{d}{\alpha} \log\left(\frac{\Gamma(\alpha)\Gamma(k\alpha - \alpha)}{\Gamma(\alpha+k\alpha - \alpha)}\right) \\
    &= \frac{d}{\alpha}\log(\Gamma(\alpha)) - \frac{d}{\alpha}\log(\Gamma(k\alpha))
\end{align}
The function $\psi(x) = \frac{d}{dx} \log(\Gamma(x))$ is commonly known as the \textit{digamma} function.

We've now shown that:
\begin{align}
    \mathbb{E}[\log g_i] = \psi(\alpha)-\psi(k\alpha)
\end{align}
Plugging into \eqref{eq:plug_back} yields the result.
\end{proof}

\section{Additional details and experiments for \textsc{Coinpress}}
\label{sec:supp_coinpress}
In Algorithms~\ref{alg:1d-mean-step} and \ref{alg:1d-mean}, we restate \textsc{Coinpress} algorithms from \cite{biswas2020coinpress} for completeness. In Algorithm~\ref{alg:1d-mean-strat}, we give a formalized version of the simple stratification algorithm discussed in the main paper body.
\begin{algorithm}
    \caption{One Step Private Improvement of Mean Interval}
    \label{alg:1d-mean-step}
    \hspace*{\algorithmicindent} \textbf{Input:} $n$ samples $X_{1\dots n}$ from $N(\mu, \sigma^2)$, $[\ell, r]$ containing $\mu$,  $\sigma^2$, $\rho_s, \beta_s > 0$ \\
    \hspace*{\algorithmicindent} \textbf{Output:} A $\rho_s$-zCDP interval $[\ell', r']$
    \begin{algorithmic}[1] 
        \Procedure{\uvmeanstep}{$X_{1\dots n}, \ell, r, \sigma^2, \rho_s, \beta_s$} 
        \State Project each $X_i$ into the interval $[\ell - \sigma\sqrt{2\log(2n/\beta_s)}, r + \sigma\sqrt{2\log(2n/\beta_s)}]$. \label{ln:mean-step-trunc}
        \State Let $\Delta = \frac{r - \ell + 2\sigma\sqrt{2\log(2n/\beta_s)}}{n}$. \label{ln:mean-step-sens}
        \State Compute $Z = \frac{1}{n}\sum_i X_i + Y$, where $Y \sim N\left(0,\left(\frac{\Delta}{\sqrt{2\rho_s}}\right)^2\right)$. \label{ln:mean-step-gm}
        \State \textbf{return} the interval $Z \pm \sqrt{2\left(\frac{\sigma^2}{n}+\left(\frac{\Delta}{\sqrt{2\rho_s}}\right)^2\right)\log(2/\beta_s)}$. \label{ln:mean-step-return}
        \EndProcedure
    \end{algorithmic}
\end{algorithm}

\begin{algorithm}
    \caption{Private Confidence-interval-based Univariate Mean Estimation}
    \label{alg:1d-mean}
    \hspace*{\algorithmicindent} \textbf{Input:} $n$ samples $X_{1 \dots n}$ from $N(\mu, \sigma^2)$, $[\ell, r]$ containing $\mu$, $\sigma^2$, $t \in \mathbb{N}^+$, $\rho_{1 \dots t}, \beta > 0$ \\
    \hspace*{\algorithmicindent} \textbf{Output:} A $(\sum_{i=1}^t \rho_i)$-zCDP estimate of $\mu$
    \begin{algorithmic}[1] 
        \Procedure{\uvmeaniter}{$X_{1 \dots n}, \ell, r, \sigma^2, t, \rho_{1 \dots t}, \beta$} 
        \State Let $\ell_0 = \ell, r_0 = r$.
         \For {$i\in [t-1]$} \label{ln:mean-est-loop}
         \State $[\ell_i, r_i] = \uvmeanstep(X_{1 \dots n}, \ell_{i-1}, r_{i-1}, \sigma^2, \rho_i, \beta/4(t-1))$.
         \EndFor
         \State $[\ell_t, r_t] = \uvmeanstep(X_{1 \dots n}, \ell_{t-1}, r_{t-1}, \sigma^2, \rho_t, \beta/4)$. \label{ln:mean-est-final}
         \State \textbf{return} the midpoint of $[\ell_i, r_i]$.
        \EndProcedure
    \end{algorithmic}
\end{algorithm}

\begin{algorithm}
    \caption{Stratified Univariate \textsc{Coinpress}}
    \label{alg:1d-mean-strat}
    \hspace*{\algorithmicindent} \textbf{Input:} Strata $X^1 \cup ... \cup X^k = X$ where each stratum $X^i \sim N(\mu_i, \sigma_i^2)$ and $X \sim N(\mu, \sigma^2)$ overall, $|X| = n$, $[\ell, r]$ containing all $\mu_i, \sigma_i$ and $\mu$, $\sigma^2$, $t \in \mathbb{N}^+$, $\rho_{1 \dots t}, \beta = \sum_{i=1}^k \beta_i > 0$, weights $\omega_1,...\omega_k$ \\
    \hspace*{\algorithmicindent} \textbf{Output:} A $(\sum_{i=1}^t \rho_i)$-zCDP estimate of global $\hat{\mu}$ as well as $(\sum_{i=1}^t \rho_i)$-zCDP estimates of $\hat{\mu_1},...,\hat{\mu_k}$ for each stratum.
    \begin{algorithmic}[1] 
        \Procedure{\stratcoin}{$X^1 \cup ... \cup X^k = X, \ell, r, \sigma^2, t, \rho_{1 \dots t}, \beta_1,...,\beta_k$, $\omega_1,...\omega_k$} 
         \For {$i\in k$} \label{stratloop}
         \State $\mu_i = \uvmeaniter(X_{1 \dots n}^i, \ell, r, \sigma^2, \rho_i, \beta_i)$.
         \EndFor
         \State $\hat{\mu} = \frac{1}{k} \sum^k_{i=1} \omega_i \mu_i$. \label{ln:mean-est-final}
         \State \textbf{return} $\hat{\mu}, \{\hat{\mu_1},...,\hat{\mu_k}\}$.
        \EndProcedure
    \end{algorithmic}
\end{algorithm}

\paragraph{Additional plots for coinpress experiments} We provide an exhaustive plot for convergence of stratified \textsc{Coinpress} variants in Figure~\ref{fig:exhaustive}.

Here we also state the generative model for the Gaussian data. We draw a Dirichlet vector $\mathbf{v} \sim \mathcal{D}(k,\alpha)$. Then, we fix uniformly random parameters for each $k$ Gaussian (in our tests, $\sigma_i \sim \mathcal{U}(\ell=0.1,u=2.0)$ and $\mu_i \sim \mathcal{N}(0,1.0)$.  Finally, using the weights of $\mathbf{v}$, we draw $k$ Gaussian samples $G_i \sim N(\mu_i,\sigma_i)$, where $|G_i| = \frac{n}{k}$. This is the data model used to produce our synthetic mean estimation results.

\begin{figure*}[ht]
\vspace{-0.3cm}
    \centering
    \includegraphics[width=\textwidth]{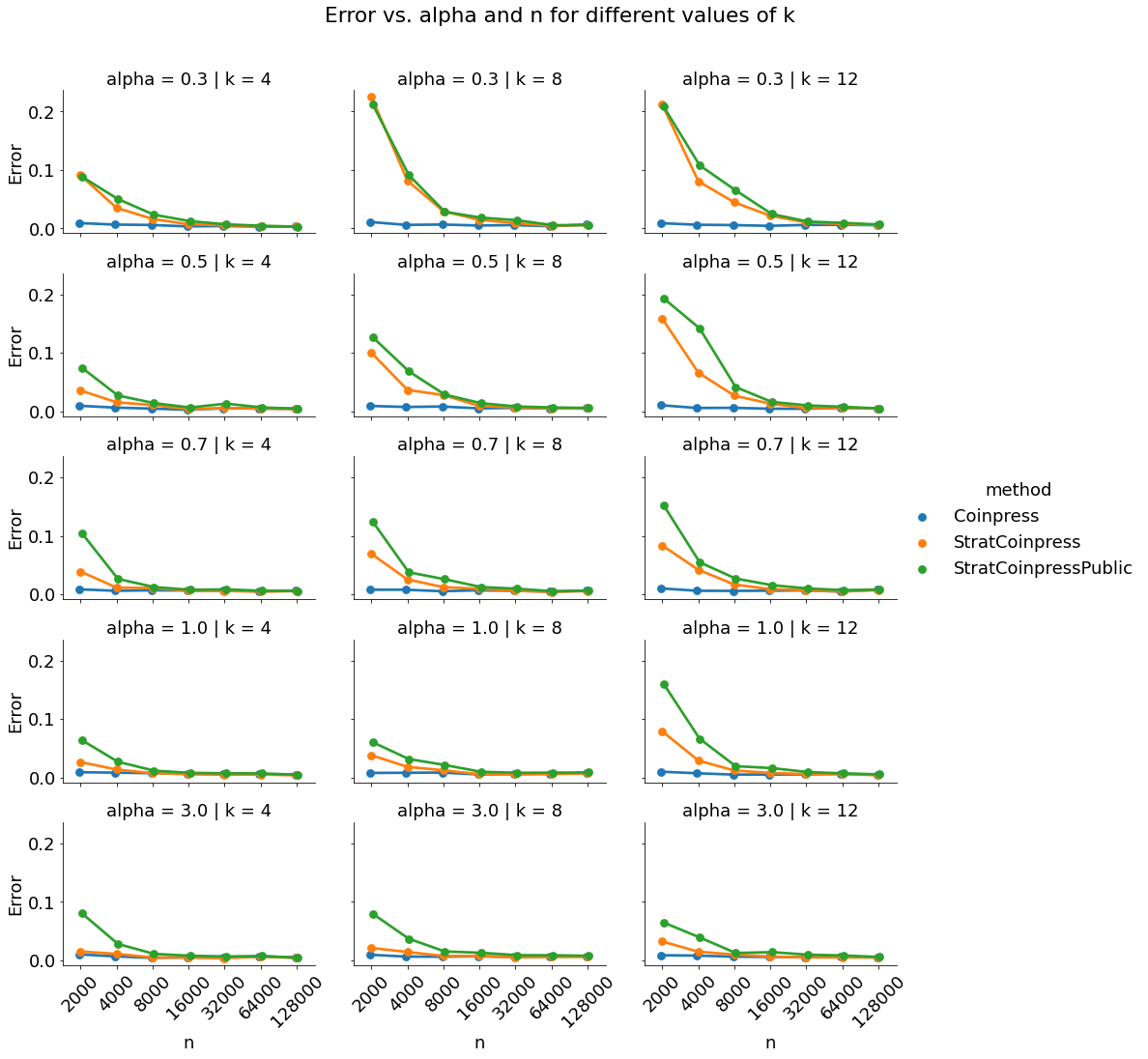}
    \caption{Complete results for synthetic \textsc{Coinpress} comparisons.}
    \label{fig:exhaustive}
\end{figure*}

\paragraph{Note on $(\rho)$-zCDP composition}
\textsc{Coinpress} relies on composition rules for zero-concentrated differential privacy.

Sequential composition: If $\mathcal{M}$ and $\mathcal{M'}$ satisfy $(\rho_1)$-zCDP and $(\rho_2)$-zCDP respectively, then the sequential $\mathcal{M}^*(X) = (\mathcal{M}(X), \mathcal{M'}(X))$ satisfies $(\rho_1 + \rho_2)$-zCDP.

Parallel composition: if $\mathcal{M}$ satisfies $\rho$-zCDP, and $\{X_1,...,X_k\}$ are disjoint, then the parallel mechanism $\mathcal{M}^*(X) = (\mathcal{M}(X_1), ..., \mathcal{M}(X_k))$ also satisfies $\rho$-zCDP. 

\section{Additional Plots and Notes on Private Data Synthesizers}
\label{sec:supp_synth_results}

\begin{figure*}[ht]
    \centering
    \centering
    \subfloat[Max disparity.]{
        \includegraphics[width=0.32\textwidth]{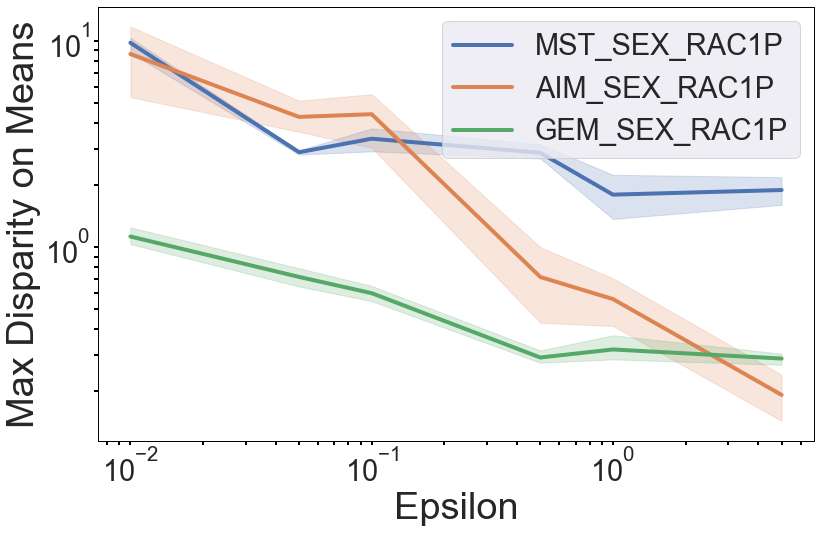}
        \label{fig:max_disp}
        
    }
    \hfill
    \subfloat[Demographic parity.]{
        \includegraphics[width=0.32\textwidth]{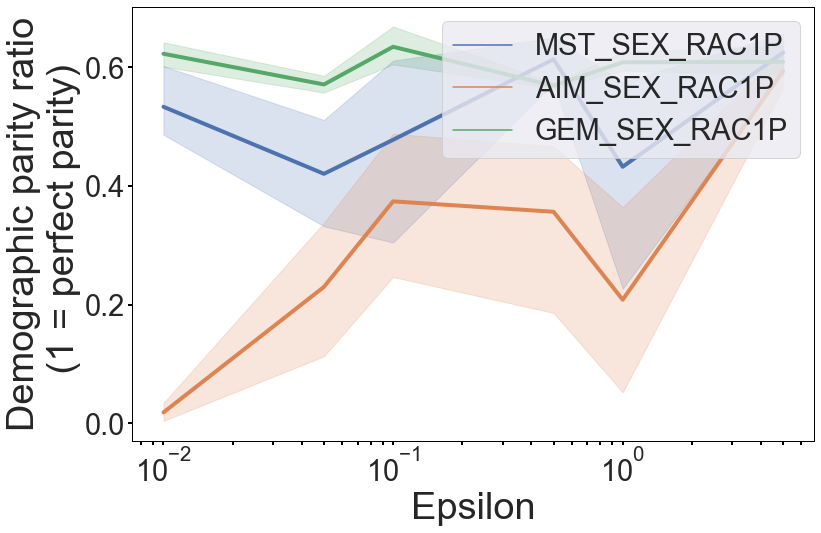}
        \label{fig:dem_par_ratio}
    }
    \hfill
    \subfloat[Max FNR Difference.]{
        \includegraphics[width=0.32\textwidth]{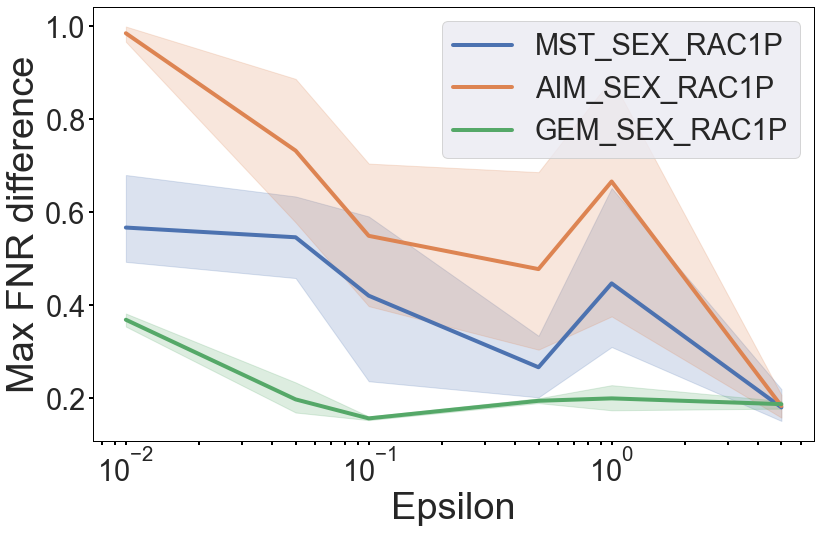}
        \label{fig:fnr_difference}
    }
    \label{fig:synthetic_results}
    
\end{figure*}

\paragraph{Discussion of challenges in analyzing stratified synthetic data} Consider a fundamental private data synthesizer like Multiplicative Weights Exponential Mechanism (MWEM) \citep{hardt2012simple}. MWEM introduced the underlying strategy of many state-of-the-art private synthesizers used today: (1) \textit{Select:} privately select a relatively small set of the most informative measurements over your data (in MWEM, this is done with the exponential mechanism \cite{dwork2014algorithmic}). (2) \textit{Measure:} privately measure that set on your real data (MWEW uses the Laplace mechanism in this step). (3) \textit{Project:} as a post-processing step, project theses private measurements onto a parameterized distribution over the domain of your data (MWEM maintains a histogram distribution and uses multiplicative weights for updates).

The challenge in analyzing stratified and non-stratified private data synthesizers lies in the initial \textit{Select} step. Were we to fix the set of selected measurements across all strata, then the aggregation step composes gracefully. In fact, in a non-private setting, the distributional weights for each variable for the entire population would simply be a weighted combination of the distributional weights for each strata. However, \textbf{we are not guaranteed that there will be overlap in measurements selected by each strata specific select step the select step run on the overall population}. For example, consider the stylized example where we limit ourselves to a single measurement in the select step. The most informative measurement candidate for the full population could be the \textit{second} most informative candidate for each stratum. In this case, we would have no direct measurements for the full population, and thus no population level utility guarantees.

\paragraph*{\textbf{Marginals}}

The following is adapted from \cite{mckenna2022aim}. They define marginals as a method to capture low-dimensional structure common in high-dimensional data distributions.  
Informally, a marginal for a set of attributes $S$ is essentially a histogram over $X_S$: it is a table that counts the number of occurrences of each $t \in \Omega_S$. Sometimes we can refer to the below functions as \emph{marginal queries}, and thus a \emph{marginal} is the resulting vector of counts $ M_S(D)$. 

\begin{definition}[Marginal \cite{mckenna2021winning,mckenna2022aim}]
Let $S \subseteq [d]$ be a subset of attributes, $\Omega_S = \prod_{i \in S} \Omega_S$, $n_S = | \Omega_S |$, and $X_S = (X_i)_{i \in S}$.  The marginal on $S$ is a vector $\mu \in \mathbb{R}^{n_S}$, indexed by domain elements $t \in \Omega_S$, such that each entry is a count, i.e., $\mu[t] = \sum_{x \in D} \mathbf{1}[x_S = t]$.  We let $M_S : \mathcal{D} \rightarrow \mathbb{R}^{n_S}$ denote the function that computes the marginal on $S$, i.e., $ \mu = M_S(D) $.  
\end{definition}

\paragraph*{\textbf{Workload Error}}
Adapted from \cite{mckenna2022aim}, the following defines a standard method for evaluating private synthetic data used to generate Figure~\ref{fig:marginals}. We define a workload as a collection of marginal queries that synthetic data should preserve. The marginal based utility measure underlying \textsc{3-Way Marginal Error} is stated in Definition~\ref{def:error}, framed as general error for any workload of queries. 

\begin{definition}[Workload Error \cite{mckenna2021winning,mckenna2022aim}] \label{def:error}
A workload $W$ consists of a list of marginal queries $S_1, \dots, S_k$ where $S_i \subseteq [d] $, together with associated weights $c_i \geq 0 $.  The error of a synthetic dataset $\hat{D}$ is defined as:
$$ \text{Error}(D, \hat{D}) = \frac{1}{k \cdot |D|}\sum_{i=1}^k c_i \left(||M_{S_i}(D) - M_{S_i}(\hat{D})||_1\right) $$
\end{definition}

\section{Additional Data Description}

\begin{table*}[]
    \centering
    \begin{tabular}{llllllllll}
        \toprule
        {} &          \textbf{SEX} &          \textbf{DIS} &     \textbf{NATIVITY} &        \textbf{RAC1P} &          \textbf{AGEP} &          \textbf{SCHL} &          \textbf{WKHP} &          \textbf{JWMNP} &                  \textbf{PINCP} \\
         &              &              &              &              &               &               &               &                &                        \\
        \midrule
        \textbf{Count}  &     196967 &     196967 &     196967 &     196967 &      196967 &      191372 &      103756 &        87355 &               166318 \\
        \textbf{Mean}   &        1.517 &        1.865 &        1.196 &        2.186 &        41.831 &         16.39 &        37.562 &         33.464 &              46657.585 \\
        \textbf{Std}    &          0.5 &        0.341 &        0.397 &        2.342 &        23.554 &         5.686 &        13.105 &         26.601 &                76027.3 \\
        \textbf{Range}  &  [1 , 2] &  [1 , 2] &  [1 , 2] &  [1 , 9] &  [0 , 95] &  [1 , 24] &  [1 , 99] &  [1 , 138] &  [-8300 , 1423000] \\
        \textbf{Median} &          2 &          2 &          1 &          1 &          42 &          18 &          40 &           30 &                26500 \\
        \bottomrule
        \end{tabular}
    \caption{This table gives descriptive information on columns from Folktables used in the mean estimation task.}
    \label{tab:mean_data}
\end{table*}

\begin{table*}[]
    \centering
    \begin{tabular}{lllllllllll}
        \toprule
        {} &         \textbf{AGEP} &         \textbf{SCHL} &          \textbf{MAR} &          \textbf{DIS} &          \textbf{ESP} &          \textbf{MIL} &         \textbf{DREM} &          \textbf{SEX} &        \textbf{RAC1P} &          \textbf{ESR} \\
         &              &              &              &              &              &              &              &              &              &              \\
        \midrule
        \textbf{Mean}   &        1.687 &        2.834 &        3.043 &        1.865 &        0.543 &        3.183 &        1.852 &        1.517 &        2.186 &        0.464 \\
        \textbf{Std}    &        1.234 &        1.246 &        1.867 &        0.341 &        1.587 &        1.547 &        0.471 &          0.5 &        2.342 &        0.499 \\
        \textbf{Range}  &  [0 , 4] &  [0 , 4] &  [1 , 5] &  [1 , 2] &  [0 , 8] &  [0 , 4] &  [0 , 2] &  [1 , 2] &  [1 , 9] &  [0 , 1] \\
        \textbf{Median} &          2 &          3 &          3 &          2 &          0 &          4 &          2 &          2 &          1 &          0 \\
        \bottomrule
        \end{tabular}
    \caption{This table gives descriptive information on columns from Folktables used in the private data synthesis task. All columns are fully represented (count $=196967$).}
    \label{tab:mean_data}
\end{table*}

\begin{table*}[]
    \centering
    \begin{tabular}{llr}
    \toprule
       \textbf{SEX} &                                              \textbf{RAC1P} &  \textbf{Count} \\
    \midrule
    Female &                                              White &  70881 \\
      Male &                                              White &  67593 \\
    Female &                                              Black &  12947 \\
      Male &                                              Black &  11077 \\
    Female &                                              Asian &   8900 \\
      Male &                                              Asian &   8130 \\
    Female &                              Some other race alone &   5668 \\
      Male &                              Some other race alone &   5296 \\
    Female &                                  Two or more races &   3026 \\
      Male &                                  Two or more races &   2620 \\
      Male &                              American Indian alone &    277 \\
    Female &                              American Indian alone &    231 \\
      Male & American Indian and/or Alaska Native (tribes no... &    125 \\
    Female & American Indian and/or Alaska Native (tribes no... &    119 \\
      Male &                                    Native Hawaiian &     42 \\
    Female &                                    Native Hawaiian &     30 \\
    Female &                                Alaska Native alone &      3 \\
      Male &                                Alaska Native alone &      2 \\
    \bottomrule
    \end{tabular}
    \caption{This table provides the exact group size breakdown for the ACSEmployment task from Folktables (New York data), used in the synthetic data experiments.}
    \label{tab:my_label}
\end{table*}

Below we also give descriptions of each Census variable used, and interpretations of their values. These are taken directly from \cite{ding2021retiring}, who in turn took them from public ACS documentation.
\begin{enumerate}
    \item AGEP (Age): Range of values:
    \begin{itemize}
        \item 0 - 99 (integers)
        \item 0 indicates less than 1 year old.
    \end{itemize}
    \item SCHL (Educational attainment): Range of values:
    \begin{itemize}
        \item N/A (less than 3 years old)
        \item 1: No schooling completed
        \item 2: Nursery school/preschool
        \item 3: Kindergarten
        \item 4: Grade 1
        \item 5: Grade 2
        \item 6: Grade 3
        \item 7: Grade 4
        \item 8: Grade 5
        \item 9: Grade 6
        \item 10: Grade 7
        \item 11: Grade 8
        \item 12: Grade 9
        \item 13: Grade 10
        \item 14: Grade 11
        \item 15: 12th Grade - no diploma
        \item 16: Regular high school diploma
        \item 17: GED or alternative credential
        \item 18: Some college but less than 1 year
        \item 19: 1 or more years of college credit but no degree
        \item 20: Associate's degree
        \item 21: Bachelor's degree
        \item 22: Master's degree
        \item 23: Professional degree beyond a bachelor's degree
        \item 24: Doctorate degree
    \end{itemize}
    \item MAR (Marital status): Range of values:
    \begin{itemize}
        \item 1: Married
        \item 2: Widowed
        \item 3: Divorced
        \item 4: Separated
        \item 5: Never married or under 15 years old
    \end{itemize}
    \item WKHP (Usual hours worked per week past 12 months): Range of values:
    \begin{itemize}
        \item N/A (less than 16 years old / did not work during the past 12 months)
        \item 1 - 98 integer valued: usual hours worked
        \item 99: 99 or more usual hours
    \end{itemize}
    \item SEX (Sex): Range of values:
    \begin{itemize}
        \item 1: Male
        \item 2: Female
    \end{itemize}
    \item RAC1P (Recoded detailed race code): Range of values:
    \begin{itemize}
        \item 1: White alone
        \item 2: Black or African American alone
        \item 3: American Indian alone
        \item 4: Alaska Native alone
        \item 5: American Indian and Alaska Native tribes specified, or American Indian or Alaska Native, not specified and no other races
        \item 6: Asian alone
        \item 7: Native Hawaiian and Other Pacific Islander alone
        \item 8: Some Other Race alone
        \item 9: Two or More Races
    \end{itemize}
    \item DIS (Disability recode): Range of values:
    \begin{itemize}
        \item 1: With a disability
        \item 2: Without a disability
    \end{itemize}
    \item ESP (Employment status of parents): Range of values:
    \begin{itemize}
        \item N/A (not own child of householder, and not child in subfamily)
        \item 1: Living with two parents: both parents in labor force
        \item 2: Living with two parents: Father only in labor force
        \item 3: Living with two parents: Mother only in labor force
        \item 4: Living with two parents: Neither parent in labor force
        \item 5: Living with father: Father in the labor force
        \item 6: Living with father: Father not in labor force
        \item 7: Living with mother: Mother in the labor force
        \item 8: Living with mother: Mother not in labor force
    \end{itemize}
    \item MIL (Military service): Range of values:
    \begin{itemize}
        \item N/A (less than 17 years old)
        \item 1: Now on active duty
        \item 2: On active duty in the past, but not now
        \item 3: Only on active duty for training in Reserves/National Guard
        \item 4: Never served in the military
    \end{itemize}
    \item DREM (Cognitive difficulty): Range of values:
    \begin{itemize}
        \item N/A (less than 5 years old)
        \item 1: Yes
        \item 2: No
    \end{itemize}
    \item PINCP (Total person's income): Range of values:
    \begin{itemize}
        \item integers between -19997 and 4209995 to indicate income in US dollars
        \item loss of \$19998 or more is coded as -19998.
        \item income of \$4209995 or more is coded as 4209995.
    \end{itemize}
    \item ESR (Employment status recode): an individual's label is 1 if ESR == 1, and 0 otherwise.
    \item JWMNP (Travel time to work): Range of values:
    \begin{itemize}
        \item N/A (not a worker or a worker that worked at home)
        \item integers 1 - 200 for minutes to get to work
        \item top-coded at 200 so values above 200 are coded as 200
    \end{itemize}
    \item WKHP (Usual hours worked per week past 12 months): Range of values:
    \begin{itemize}
        \item N/A (less than 16 years old / did not work during the past 12 months)
        \item 1 - 98 integer valued: usual hours worked
        \item 99: 99 or more usual hours
    \end{itemize}
\end{enumerate}

\begin{figure*}
    \centering
    \includegraphics[width=0.45\textwidth]{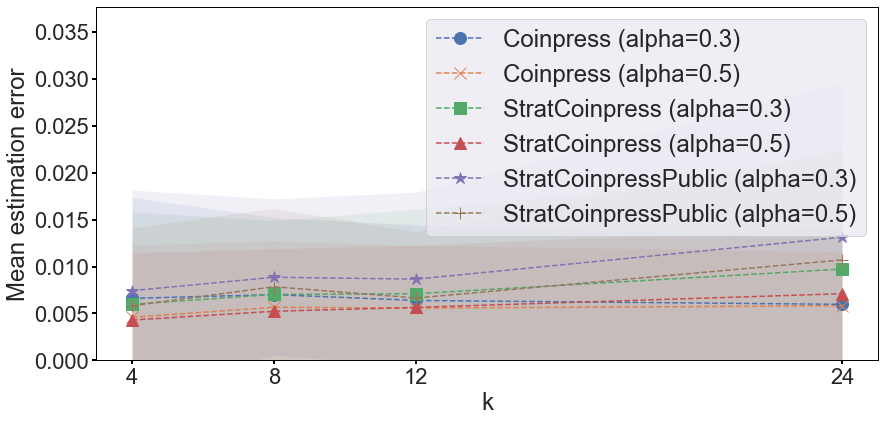}
    \caption{Comparing mean estimation error of \textsc{Coinpress} and \textsc{StratCoinpress}, over 50 runs. Mean estimation error, varying the number of groups $k$ in our synthetic Gaussian mixture. Note that error is low and relatively stable across reasonable values of $k$ groups in U.S. Census data.}
    \label{fig:error_vs_k}
\end{figure*}

\begin{figure*}[ht]
    \centering
    \subfloat{
        \includegraphics[width=\textwidth]{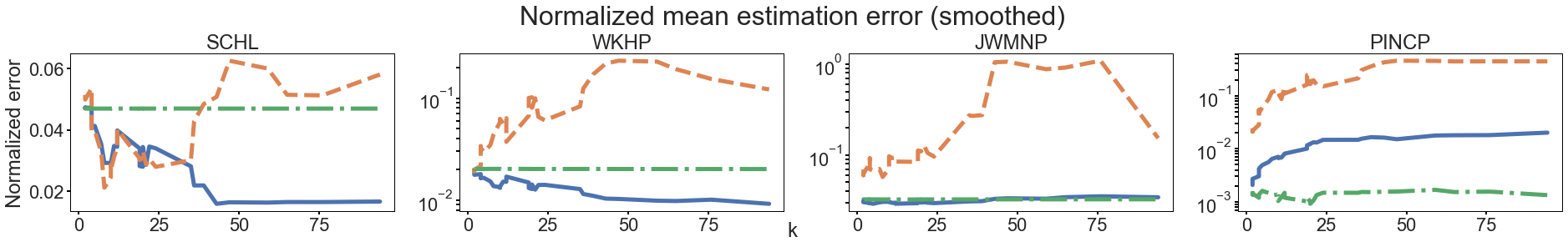}
        \label{fig:coinpress_mean_error}
    }
    \hfill
    \subfloat{
        \includegraphics[width=\textwidth]{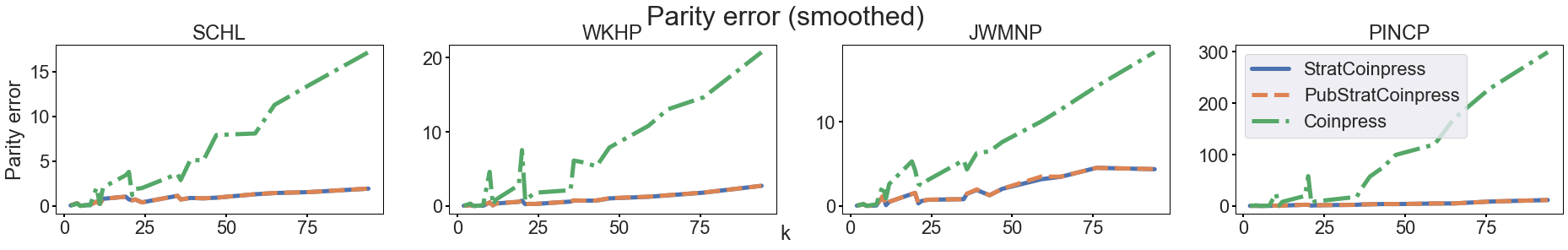}
        \label{fig:coinpress_parity_error}
    }
    \caption{(Full version of~\ref{fig:folktables_coinpress_results}) Comparisons of \textsc{Coinpress} and \textsc{StratCoinpress} on U.S. Census data from Folktables for the Mobility task \cite{ding2021retiring}; $\texttt{SCHL}$ is years spent in school (ordinal, 0-24), $\texttt{WKHP}$ is hours worked per week (ordinal, 0-168), $\texttt{JWMNP}$ is average drive time (ordinal, 0-400) and $\texttt{PINCP}$ is income (continuous). Top row shows normalized mean estimation error, bottom row shows parity error, both as the number of groups $k$ increases.}
    \label{fig:folktables_coinpress_results_full}
\end{figure*}

\begin{figure*}
\centering
\caption{Parity error, where $f$ is the 3-way marginal function \cite{mckenna2022aim}. On left, error when strata are just the different $race$ groups. On right, error for intersectional $sex$ and  $race$ groups.}
\includegraphics[width=0.7\textwidth]{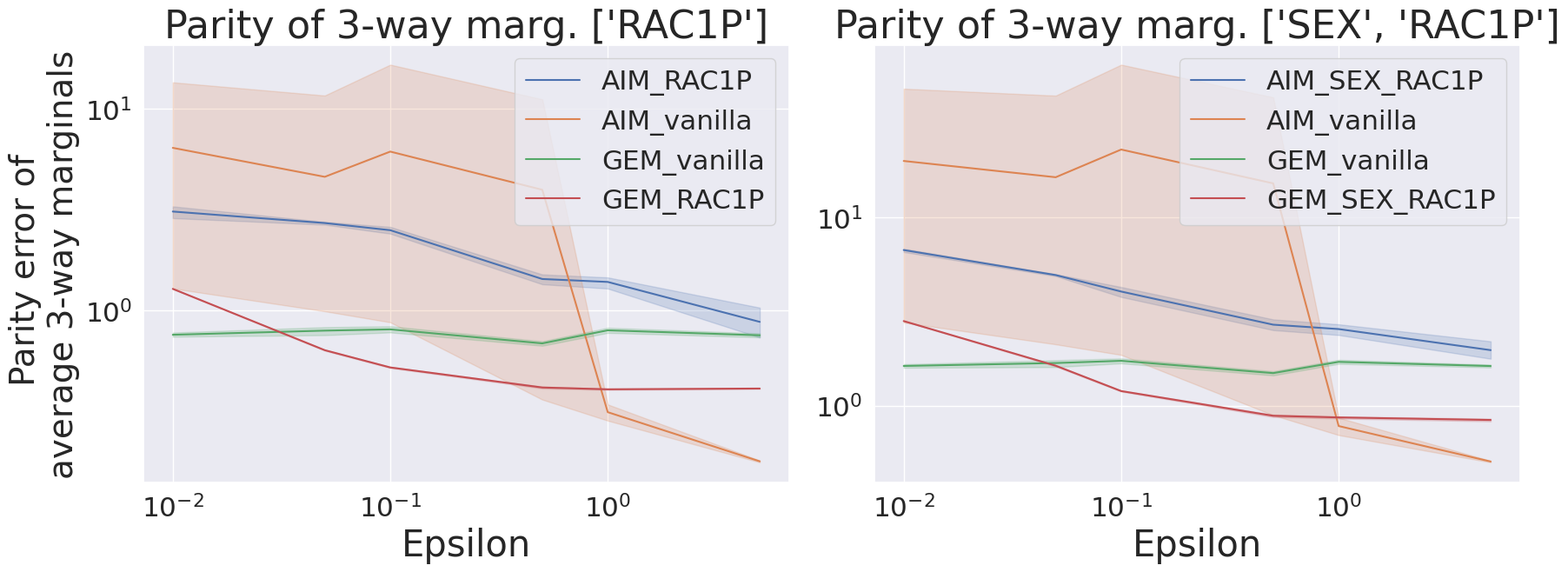}
\label{fig:marginals}
\end{figure*}
\newpage
\begin{figure*}[]
    \centering
    \subfloat{\includegraphics[width=.33\textwidth]{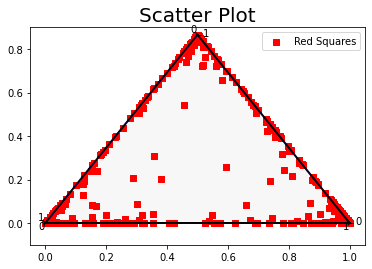}}
    \subfloat{\includegraphics[width=.33\textwidth]{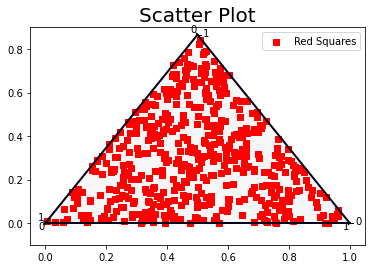}}
    \subfloat{\includegraphics[width=.33\textwidth]{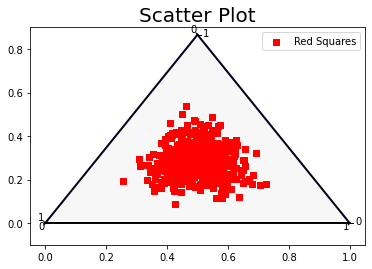}}
    \caption{$k$=3 random vectors, $ \mathcal{D}(\alpha=\{0.1,1,10\})$}
    \label{fig:dir}
\end{figure*}
\section{The Dirichlet distribution}
\label{sec:supp_notes}
A Dirichlet distribution ($\mathcal{D}$) has a probability density function given by the following, where $k \geq 2$, $\alpha_i > 0$, and $B$ is a beta function expressed in terms of a gamma function
\begin{align}
    f(x_1,...,x_k; \alpha_1,...,\alpha_k) = \frac{1}{B(\{\alpha_i\})}\prod_{i=1}^k x_i^{\alpha_i -1} \\
    B(\{\alpha_i\}) = \frac{\prod_{i=1}^k \Gamma(\alpha_i)}{\Gamma(\sum^k \alpha_i)}
\end{align}
A special case, and the one we primarily consider here, is the symmetric Dirichlet distribution, when $\alpha_1,...,\alpha_k=\alpha$.
\begin{align}
    f(x_1,...,x_k; \alpha) =  \frac{\Gamma(\alpha k)}{\Gamma( \alpha)^k}
    \prod_{i=1}^k x_i^{\alpha -1}
\end{align}
Canonically, the Dirichlet distribution is known as the ``distribution over distributions,'' and is used to draw a random vector of size $k$, where the entries of that vector sum to 1. The distribution of those entries is governed, in the symmetric case, by a single parameter $\alpha$. When $\alpha < 1$, the magnitude of the $k$ entries in the vector are distributed less evenly, and when $\alpha > 1$ we see a more even distribution. We can visualize this nicely when $k=3$ by drawing a set of 100 random dirichlet vectors and plotting them in Figure~\ref{fig:dir}.

\end{document}

%% file: ack.tex
\section{Acknowledgments}
\label{sec:ack}

This research was supported in part by NSF Award Nos. 1916505, 2312930, 1922658, by NSF Award No. 2045590 and by the NSF Graduate Research Fellowship under Award No. DGE-2234660.